\documentclass[a4paper,aps,pra,showpacs,twocolumn]{revtex4-1}      
\usepackage[utf8]{inputenc}  
\usepackage[T1]{fontenc}     
\usepackage[british]{babel}  
\usepackage[sc,osf]{mathpazo}
\usepackage[scaled=0.86]{berasans}  
\usepackage[scaled=1.03]{inconsolata} 
\usepackage[usenames,dvipsnames]{color} 
\usepackage[x11names]{xcolor} 
\usepackage[colorlinks,citecolor=magenta,linkcolor=SpringGreen4,urlcolor=blue]{hyperref}  
\usepackage{graphicx} 
\usepackage[babel]{microtype}  
\usepackage{amsmath,amssymb,amsthm,bm,amsfonts,mathrsfs,bbm} 
\usepackage{xspace}  
\usepackage{pgfplots}  
\DeclareMathOperator{\tr}{tr}

\newcommand{\ket}[1]{\left| #1 \right\rangle}

\newcommand{\ketbra}[2]{\left|#1\middle\rangle\middle\langle#2\right|}

\newcommand{\eg}{\emph{e.g.}\@\xspace}
\newcommand{\ie}{\emph{i.e.}\@\xspace}

\newtheorem{theorem}{Theorem}
\newtheorem{lemma}{Lemma}

\newtheorem{corollary}{Corollary}

\newcommand{\ba}{\begin{eqnarray}}
\newcommand{\ea}{\end{eqnarray}}
\newcommand{\ban}{\begin{eqnarray*}}
\newcommand{\ean}{\end{eqnarray*}}

\begin{document}

\title{Super-Activation of Quantum Steering}

\author{Marco T\'ulio Quintino}
\affiliation{Groupe de Physique Appliqu\'ee, Universit\'e de Gen\`eve, CH-1211 Gen\`eve, Switzerland}
\author{Marcus Huber}
\affiliation{Institute for Quantum Optics and Quantum Information (IQOQI), Austrian Academy of Sciences, Boltzmanngasse 3, 
A-1090 Vienna, Austria}
\author{Nicolas Brunner}
\affiliation{Groupe de Physique Appliqu\'ee, Universit\'e de Gen\`eve, CH-1211 Gen\`eve, Switzerland}

\date{\today}  

\begin{abstract}
We consider Einstein-Podolsky-Rosen steering in the regime where the parties can perform collective measurements on many copies of a given shared entangled state. We derive a simple and efficient condition for guaranteeing that an entangled state is $k$-copy steerable. In particular we show that any two-qubit and qubit-qutrit entangled states is $k$-copy steerable. This allows us to discuss the effect of super-activation of steering, whereby an entangled state that is unsteerable (\ie, admits a local hidden state model) in the one-copy regime becomes $k$-copy steerable. We provide examples with few copies and low dimensions. Our results give evidence that entanglement and steering could become equivalent in the multi-copy regime. 
\end{abstract}

\maketitle

\section{Introduction}

The notion of nonlocality appears in quantum mechanics in several forms. Bell nonlocality \cite{bell64,brunner_review} is the strongest form of inseparability and is relevant to device-independent quantum information protocols \cite{acin07}. Two distant observers can observe quantum nonlocality by performing local measurements on a shared entangled state, in the sense that the resulting statistics cannot arise from a local hidden variable (LHV) model. Another concept is Einstein-Podolsky-Rosen steering \cite{einstein35,wiseman07}, or simply quantum steering, a form of nonlocality that is intermediate between entanglement and Bell nonlocality. Here the measurement statistics are tested with respect to a local hidden state (LHS) model, which represents a particular class of LHV models where the hidden variable can be considered as a quantum state. Steering also finds applications in quantum information \cite{branciard11b}.
	
	Interestingly these different forms of non separability are strictly different at the level of quantum states \cite{werner89,barrett02}. Specifically, there exist entangled states which cannot lead to steering, \ie, admitting a LHS model \cite{wiseman07}. Furthermore, there exist entangled states that can lead to steering but cannot give Bell nonlocality, \ie, admit a LHV model \cite{wiseman07}. In fact, there is a strict hierarchy between Bell Nonlocality, steering, and entanglement \cite{quintino15}, even when considering the most general measurements, \ie, positive-operator-valued-measures (POVMs).

	\begin{figure}[b!]
		\includegraphics[width=0.9 \columnwidth]{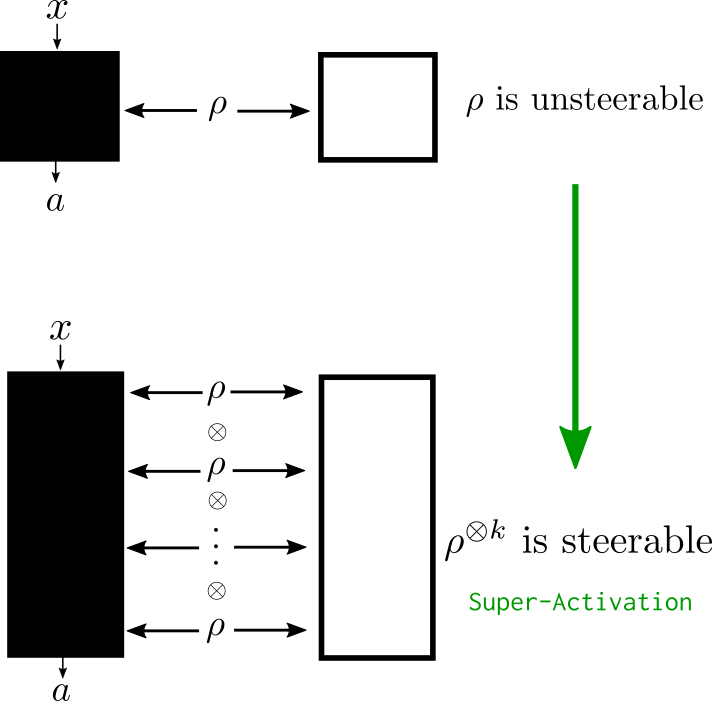}
		\caption{Two parties can obtain quantum steering by performing local measurements on $k$ copies of an entangled state that cannot lead to steering for a single copy. This phenomenon is referred to as the super-activation of steering. We consider the situation where Alice steers Bob. That is, Alice is the untrusted part (represented by a black box) while Bob is the trusted part characterizing his assemblage via quantum tomography (represented by a white box).}
\end{figure}

	These results were established in the so-called standard Bell/steering scenario, where the parties perform local measurements on a single copy of the entangled state $\rho$ in each round of the test. In order to obtain the statistics for testing a Bell or steering inequality, many rounds are required. However, one may consider more general Bell/steering tests. In particular, we will consider here the situation where the parties use $k$ copies of the entangled state $\rho$, which can be described as a global bipartite entangled state $\rho^{\otimes k}$, and perform local collective measurements on it. This will be referred to as the $k$-copy Bell/steering scenario. 
	
	Importantly, the parties can now perform localy some joint measurements (\eg, featuring entangled POVM elements) on $k$ local systems. This leads to novel possibilities. Namely, it can be the case that an entangled state leads to Bell nonlocality or steering in the $k$-copy scenario (\ie, $\rho^{\otimes k}$ is nonlocal/steerable), despite the fact the initial entangled state $\rho$ admits a LHV/LHS model. This phenomenon is called super-activation of nonlocality/steering. This is possible because the statistics of the joint measurement on $k$ copies of $\rho$ are not covered by the LHV/LHS model for $\rho$ which considers only measurements on a single copy.


In a seminal paper, Peres considered this question, however in a slightly more general scenario, where the parties can perform an initial pre-processing on many copies of the state \cite{peres96} (see also \cite{masanes06}). In this case, any entangled state that is distillable becomes nonlocal. The scenario without pre-processing, which we focus on here, was only considered later. First, Liang and Doherty \cite{liang06} showed that the maximal violation of the CHSH inequality \cite{chsh69} of a particular state $\rho$ can be increased if the parties perform joint measurements on multiple copies. A next step was made by Navascués and Vértesi \cite{vertesi11} who presented entangled states such that $\rho$ does not violate the CHSH inequality but $\rho\otimes\rho$ does. These results demonstrate the activation of quantum nonlocality, which was also discussed in the multipartite case \cite{caban15,paul15}. Other works showed the activation of general nonlocal non-signaling correlations via local wirings, a process termed nonlocality distillation \cite{forster09,BS,brunner10}.

An even stronger form of activation was discovered by Palazuelos, who showed that quantum nonlocality can be super-activated \cite{palazuelos12}. Specifically, he proved that certain entangled states $\rho$ admitting a LHV model for general POVMs can be super-activated, in the sense that $\rho^{\otimes k}$ violates a Bell inequality (for some finite $k$). To derive this result, he took advantage of known Bell inequalities with unbounded quantum violation \cite{junge09,KV,buhrman10}. Next, Cavalcanti and colleagues \cite{cavalcanti12} presented a general criterion for $k$-copy nonlocality. Specifically any entangled state $\rho$ of dimension $d \times d$ that is useful for quantum teleportation (\ie, with maximal entanglement fraction greater than $1/d$ \cite{horodecki98}) is $k$-copy nonlocal. This construction provides large classes of entangled state for which nonlocality can be super-activated. In particular this is the case for any isotropic entangled state (of any dimension) which admits a LHV model.

In this paper we discuss the $k$-copy scenario for quantum steering. Our main result is a simple and general criterion for $k$-copy steerability of a quantum state $\rho$. Our method exploits the reduction criterion \cite{horodecki_reduction}, used in the context of entanglement distillation, and can thus be computed efficiently, by finding the minimal eigenvalue of an operator. This result implies that any two-qubit and any qubit-qutrit entangled state is $k$-copy steerable. Therefore, given previous results constructing LHS models for such entangled states \cite{wiseman07,bowles13,sania14,bowles16,hirsch15,cavalcanti15,nguyen16}, we obtain many examples of super-activation of steering. Our results also allows us to put bounds on the minimal number of copies required for super-activation of steering and nonlocality. We discuss the connection of our results to entanglement distillation and conclude with some open questions. More generally, our results suggest that $k$-copy steerability may be generic for all entangled states (we prove it for $2 \times 2$ and $2 \times 3$ entangled states). This would demonstrate the equivalence of entanglement and steering in the multi copy scenario.

\section{Bell Nonlocality and Quantum Steering}

 Consider two distant observers, Alice and Bob, that can perform local measurements on a shared entangled quantum state $\rho$.	 			%
 	   As standard, quantum measurements are described by the set of POVMs  $\{  A_{a | x}\}$ and $\{B_{b  |y}\}$, positive operators that sum to identity,  where $  x$ and $ y $  represent the choice of the measurements and  $a$ and $b$ their corresponding  outcomes. The resulting statistics are given by
\begin{equation}
p(ab| xy) = \tr(\rho A_{a\vert x} \otimes B_{b| y})   .
\end{equation}
The above distribution is Bell local when it admits a decomposition of the form%
\begin{equation} \label{bell}
		p(ab   \vert xy) = \int   \pi(\lambda) p_A(a\vert x,\lambda)p_B(b\vert y,\lambda) \mathrm{d} \lambda ,
\end{equation}
	where $\lambda $ is some classical random variable, distributed according to density $\pi(\lambda)$, and ${p_A(a\vert x,\lambda) }$ and $p_B(b\vert y,\lambda)$ represent the local response functions. A quantum state $\rho$ is said to be Bell local, or equivalently to admit a LHV model, if its statistics admit a decomposition of the form of Eq. \ref{bell} for \textit{all} possible measurements. %
	If such a decomposition does not
	 exist%
	 ,
	  the state is Bell nonlocal and violates at least one Bell inequality for some choice of local measurements $\{A_{a|x}\}$ and $\{B_{b|y}\}$ \cite{bell64,brunner_review}.

	A different notion of nonlocality is that of steering. In a steering test, Bob, who does not trust Alice, wants to verify that a state $\rho$ is entangled. To this end, he asks Alice to perform measurement $x$ on her subsystem and announce its result $a$. By doing so, she remotely steers the state of Bob's system to
			\begin{equation}\label{assemblage}
				\sigma_{a | x}=\tr_A(A_{a | x}\otimes I \rho),
			\end{equation}
where $\tr_A$ denotes the partial trace over Alice's system. Bob's task is now to ensure that the set of unnormalized conditional states $\{\sigma_{a\vert x}\}$, a so-called \textit{assemblage}, does not admit a decomposition of the form 
\begin{equation} \label{uns}
\sigma_{a| x}=  \int \pi(\lambda) p_A(a  |x, \lambda)\rho_\lambda \mathrm{d} \lambda,
\end{equation}
where $\lambda$ is a classical random variable distributed according to $\pi(\lambda)$,  $p_A(a\vert x, \lambda)$ represents Alice's response function, and $\rho_\lambda$ the (hidden) quantum state. If the assemblage observed by Bob (that can be obtained via quantum tomography%
) does admit a decomposition \eqref{uns}, Bob concludes that Alice could have cheated by using the following strategy: Alice sends the single party, hence not entangled, quantum state $\rho_\lambda$ to Bob, and announces the measurement outcome $a$ according to $p_A(a\vert x, \lambda)$. 
Note that $\lambda$ can be understood as Alice's choice of strategy. 
If an entangled state $\rho$ admits a decomposition of the form of Eq. \eqref{uns} for quantum measurements, the state is termed as unsteerable (or equivalently, it admits a LHS model).
 However, if the assemblage $\{\sigma_{a  |x}\}$ does not admit a decomposition of the form \eqref{uns} for some set of measurements $A_{a | x}$, the state is said to be steerable and it violates a steering inequality \cite{cavalcanti09}.

We note that a LHS model is a particular case of a LHV model in which Bob's response functions respects the trace rule imposed by quantum mechanics
(see \cite{augusiak_review}).
 Hence, states admitting a LHS model always admit a LHV one, whereas the converse does not hold in general \cite{quintino15}. We also point out that the asymmetry of steering allows the existence of entangled states which are only one-way steerable \cite{bowles14,quintino15}.

In the next sections, we will make use of a quantification of bipartite quantum entanglement. The \textit{maximal entanglement fraction} \cite{horodecki98} $F$ of a state $\rho$ is defined as the maximal overlap of $\rho$ with any maximally entangled state. More precisely,
\begin{equation}
F(\rho):= \max_{U_A,U_B} \tr(\rho U_A\otimes U_B \ketbra{\phi_d^+}{\phi_d^+}U_A^\dagger \otimes U_B^\dagger),
\end{equation}
where $U_A$ and $U_B$ are local unitary operators and $\ket{\phi_d^+}= \frac{1}{\sqrt{d}} \sum_i\ket{ii}$ the $d$-dimensional maximally entangled state.

The concept of maximal entanglement fraction is intimately related to \textit{isotropic states} \cite{horodecki_reduction}, quantum states with a high level of symmetry. A $d$-dimensional isotropic state with entanglement fraction $F$ is defined as 
		\begin{equation}
		ISO_d(F)=F \ketbra{\phi^+_d}{\phi^+_d} + (1-F) \frac{I_{d\times d}-\ketbra{\phi^+_d}{\phi^+_d}}{d^2-1}.
		\end{equation}	
	It can be shown that any quantum state $\rho$ acting in $\mathbb{C}^d\otimes \mathbb{C}^d$ with maximal entanglement fraction $F$ can be transformed into a $d$-dimensional isotropic state with entanglement fraction $F$ by local unitary operations and shared randomness via an \textit{isotropic twirling} operation \cite{horodecki_reduction}. Since local unitary transformation and shared randomness cannot create nonlocality (steering nor Bell), it follows that if a state $\rho$ can be transformed into a steerable (Bell nonlocal) isotropic state, $\rho$ is guaranteed to be steerable (Bell nonlocal).

\section{Steerability of multiple copies}

Our goal now is to discuss steering in the $k$-copy scenario. Specifically, given an initial entangled state $\rho$, we will consider the assemblage 
\begin{equation}\label{assemblage2}
				\sigma_{a | x}=\tr_A( A_{a | x}\otimes I \rho^{\otimes k}),
			\end{equation}
			where $A_{a | x}$ now represents a local measurement of Alice performed on all $k$ subsystems she holds. In particular, these can be any possible joint measurement. In case one can find a set of measurements $A_{a | x}$ such that the resulting assemblage is steerable, we say that the state $\rho$ is $k$-copy steerable (from Alice to Bob).

	 We now present our main result, which is a simple and general criterion for ensuring that an entangled state $\rho$ is $k$-copy
	 steerable. In particular, this result then implies that $k$-copy steerability is generic for two-qubit and qubit-qutrit entangled states.

\begin{theorem} \label{super}
	All quantum states $\rho$ acting on $\mathbb{C}^d\otimes \mathbb{C}^d$ for which the operator \begin{equation}
	I_d\otimes\rho_B-\rho 
	\end{equation}%
	is not positive definite (\ie, $\rho$ violates the reduction criterion) are $k$-copy steerable from Alice to Bob for some number of copies $k$.
	
	Moreover, let $F$ be the maximal entanglement fraction of $\rho$. Then, $\rho$ is $k$-copy steerable whenever
		\begin{equation}  \label{alternative_criterion}
			F^k> \left[(1+d^k)\left(\sum_{i=1}^{d^k}\frac{1}{i}-1\right)-d^k\right] \frac{1}{d^{2k}	}.
		\end{equation}
\end{theorem}
		
		We note that the \textit{reduction criterion} \cite{horodecki_reduction,guhne08} gives a sufficient condition for entanglement distillability. The connection to this topic will be discussed in Section \ref{distillation}.
		
\begin{proof}
	
		The first step of the proof makes use of properties of assemblages admitting a LHS model. Specifically, we recall that local unitary operations, shared randomness, and local filtering operations on Bob's side preserve unsteerability \cite{quintino15,aolita14,bowles16}. 
		More precisely, one has the following statement. 
		
		\begin{lemma}
		Let $\rho$ be an entangled state, unsteerable from Alice to Bob.
For any local operation represented by a positive map $\Lambda$ acting on Bob's system, the final
state
\begin{equation}\label{filtered}
\rho_{F}=\frac{(I\otimes\Lambda)(\rho)}{\tr[(I\otimes\Lambda)(\rho)]}
\end{equation}
is unsteerable from Alice to Bob.
\end{lemma}
		
		We refer the reader to Refs \cite{quintino15,bowles16} for the proof.

The second step consists in using the reduction criterion. Note that any entangled state $\rho$ for which the operator $		I_d\otimes\rho_B-\rho $ is not positive definite, can be transformed into a $d$-dimensional isotropic state $ISO_d(F)$ with $F>1/d$ via an isotropic twirling operation, which consists of local unitaries and shared randomness \cite{horodecki_reduction}. 
Therefore, for any state $\rho$ violating the reduction criterion, we can construct a local filter on Bob's side that will map $\rho$ to an isotropic state $ISO_d(F)$ with $F>1/d$ (after isotropic twirling); see Appendix for details on the local filter. Since the latter is $k$-copy Bell nonlocal, following the result of Ref. \cite{cavalcanti12}, it follows that it is also $k$-copy steerable; this is because Bell nonlocality always implies steering. Therefore, it follows from Lemma 1 that $\rho$ is $k$-copy steerable.

		We also give a more refined proof that requires in general less copies of the initial state $\rho$. We
		observe that $k$-copies of a $d$-dimensional isotropic state have the form
		\begin{align}
		\left[ISO_d(F)\right]^{\otimes k} =& F^k \ketbra{\phi^+_d}{\phi^+_d}^{\otimes k} + (1-F^k) \rho_R,\\
		=& F^k \ketbra{\phi^+_{d^k}}{\phi^+_{d^k}}  + (1-F^k) \rho_R ,
		\end{align}
		where the state $\rho_R$ is orthogonal to the maximally entangled one, \ie, $\tr \left(\ketbra{\phi^+_{d^k}}{\phi^+_{d^k}} \rho_R\right)=0$. 
		Hence, we can transform $k$ copies of a $d$-dimensional isotropic state with entanglement fraction $F$ into a single copy of a $d^k$-dimensional isotropic state with entanglement fraction $F^k$ via an isotropic twirling operation, that is,
		\begin{equation}
		\left[ISO_d(F)\right]^{\otimes k} \mapsto 	ISO_{d^k} (F^k).
		\end{equation}
	
	Next, we exploit a result given in Ref. \cite{wiseman07}, where the authors presented a necessary and sufficient condition for steerability of a $d$-dimensional isotropic state, considering all projective measurements. This criterion can be expressed in terms of the maximal entanglement fraction. Specifically, the state is steerable if and only if
	\begin{equation}
	F> F_{proj}=\frac{\left[(1+d) \sum_{n=1}^d \frac{1}{n} \right]-d} {d^2}.
	\end{equation} 
	Applied to our case, we obtain a sufficient condition for $k$-copy steerability  of the initial state. More precisely, any state $\rho$ with maximal entanglement fraction $F$ is $k$-copy steerable whenever
			\begin{equation}  
			F^k> \frac{(1+d^k)(\sum_{i=1}^{d^k}\frac{1}{i}-1)-d^k }{d^{2k}	}.
			\end{equation}
\end{proof}

	The above result leads to a general result for all qubit-qubit and qubit-qutrit entangled states, as all these states violate the reduction criterion \cite{horodecki_reduction,horodecki_ppt}. Hence we have that
\begin{corollary}
		All entangled $2 \times 2$ entangled state are $k$-copy (two-way) steerable, and all $2 \times 3$ entangled states are $k$-copy steerable (from the qubit to the qutrit part).
\end{corollary}

	\section{Super-Activation of Quantum Steering}

	In the previous sections we have discussed how to reveal steering of a quantum state $\rho$ by performing measurements on many copies copies of it. We are thus now all set up to address the question of super-activation. That is, we look for states which admit a LHS model in the single copy regime but violate a steering inequality when $k$ copies are considered. More precisely we focus on the question: ``\emph{Can we find a state $\rho$ which admits a LHS model, such that $\rho^{\otimes k}$ becomes steerable?}''
	
Already the results of the previous section can directly be employed to reveal super-activation of quantum steering. First of all, the fact that every entangled $2\times 2$ state is steerable for some number of copies $k$, together with the fact that there are numerous $2 \times 2$ entangled state admitting a LHS model \cite{wiseman07,bowles13,sania14,bowles16,hirsch15,cavalcanti15,nguyen16} is sufficient to show super-activation. The same is true of course for the case of $2 \times 3$ entangled states admitting a LHS model (from Alice to Bob). More generally, any entangled state admitting a LHS and violating the reduction criterion will feature super-activation of steering. 

Moreover, we note that the mere existence of the super-activation of steering already follows from the works of Refs \cite{palazuelos12} and \cite{cavalcanti12}. This is due to the fact that (i) $k$-copy nonlocality implies $k$-copy steerability (two-way), and (ii) some of the entangled states discussed in Refs \cite{palazuelos12,cavalcanti12} admit a LHS model. Nevertheless, the present results are stronger, in the sense that our criterion can detect much larger classes of entangled states compared to Refs \cite{palazuelos12,cavalcanti12}. For instance, our criterion detects all two-qubit entangled states, which is not the case for Refs \cite{palazuelos12,cavalcanti12}. Moreover, our results also provide much stronger bounds on the minimal number of copies $k$ required for super-activation as we will discuss now.

\subsection{Few copies and low dimensions}

While we have discussed above several examples of the super-activation of steering, we ask now a more sophisticated questions, looking for minimal examples of super-activation. Specifically, we ask the question: ``\emph{Let $\rho$ be a state that can be super-activated. What is the minimum number of copies $k$ for the final state  $\rho^{\otimes k}$ to be steerable?}''
Below we give example of super-activation of steering with few copies and construct an example where $k=2$ copies are enough. Moreover, we also give examples for states of low dimension.

	We start with the $d$-dimensional isotropic state with entanglement fraction $F$ 
	\begin{equation}
		ISO_d=F \ketbra{\phi^+_d}{\phi^+_d} + (1-F) \frac{I_{d\times d}-\ketbra{\phi^+_d}{\phi^+_d}}{d^2-1}
	\end{equation}
 that	is known to have an LHS model for projective measurements if and only if \cite{wiseman07} 
		\begin{equation}\label{proj}
		F\leq \frac{\left[(1+d) \sum_{n=1}^d \frac{1}{n} \right]-d} {d^2},
		\end{equation}
		and to have a LHS model for POVMs if \cite{almeida07}
	\begin{equation}\label{povm}
	 F\leq \frac{1+\left(\frac{d+1}{d}\right)^d(3d-1)}{d^2}.
	\end{equation}
	Using theorem \ref{super} one can explicitly find sufficient requirements on the dimension $d$ and number of copies $k$ for obtaining super-activation of steering. We start with the case $d=2$. From the above inequalities \eqref{proj} and \eqref{povm} we see that if Alice and Bob share $k$ copies of a two-qubit isotropic state, they can super-activate steering for projective measurements with $k=7$ copies, and super-activate steering for POVMs with $k=24$ copies. Next we discuss the case of only $k=2$ copies. Here, considering projective measurements, we get that steering is super-activated for isotropic states of dimension $d \geq 6$. However, when considering POVMs, the combination of theorem \ref{super} and condition \eqref{povm} does not allows us to prove super-activation directly.

 Moreover, we remark that the criterion for $k$-copies steerability presented in equation \eqref{alternative_criterion} is always more ``economical'' in number of copies than the criterion presented in Ref. \cite{cavalcanti12}. This is due to the fact that both proof methods rely on the nonlocality of an isotropic state, but while Ref. \cite{cavalcanti12} focuses on a specific Bell inequality, namely the Khot-Vishonoi game \cite{KV,buhrman10}, our method considers \textit{all} steering inequalities for projective measurements via the necessary and sufficient criterion presented in \cite{wiseman07}.

	Before finishing this section we point out a simple method to show super-activation of steering (or Bell nonlocality) for the scenario where the parties share only two copies of an initial state $\rho$. That is, there exits $\rho$ which has a LHS (LHV) model for general POVMs and $\rho\otimes \rho$ violates a steering (Bell) inequality. While this method is implicitly suggested in Palazuelos' paper \cite{palazuelos12}, we present it here in full detail for completeness 
	\begin{enumerate}
		\item Choose $\rho$ acting on $\mathbb{C}^d\otimes\mathbb{C}^d$ that has a LHS (LHV) model for POVMs, and that is $k$-copy steerable (Bell nonlocal).
		\item There exists a critical number of copies $k_C>1$ such that $\rho^{\otimes k_c}$ is unsteerable (Bell local), but $\rho^{\otimes k_c+1}$ is steerable (Bell nonlocal)
		\item By construction, the new state $\rho'=\rho^{\otimes k_c}$ is unsteerable (Bell local) but two copies of it violate a steering (Bell) inequality.
	\end{enumerate}
	Following the above protocol, one can transform any example of $k$-copy super-activation of steering (Bell nonlocality) into another one that requires only 2 copies. The ``price'' to pay for this construction is that the dimension of the quantum state $\rho'=\rho^{\otimes k_c}$ that exhibits super-activation with two copies becomes larger than the initial one.
	
	\section{Connection to Entanglement Distillation} \label{distillation}
	
	Two parties sharing many copies of an entangled state $\rho$ and having access to local operations and classical communication (LOCC) can in some cases obtain (or more precisely, arbitrarily approximate) a maximally entangled state via an \textit{entanglement distillation protocol} \cite{bennett96,horodecki_review}.

	Clearly, entanglement distillation shares similarities with the scenario of $k$-copy steering. The notable difference is however that LOCC is not considered in the latter. Thus, the fact that a given state is distillable does not necessarily imply that it will lead to $k$-copy steering, since LOCC is not considered in $k$-copy steering.

	More generally, the link between entanglement distillation and nonlocality in asymptotic scenarios has been discussed. First, it was conjectured by Peres \cite{peres96} that all entangled but undistillable states, \ie, bound entangled states, admit a LHV model Bell even in the $k$-copy regime with local pre-processing. This conjecture was further strengthened to the case of an LHS model \cite{pusey13}. Nevertheless, these conjectures are known to be false since there exist bound entangled states that can lead to steering \cite{moroder14} and Bell nonlocality \cite{vertesi14}. In fact, these examples consider the standard single copy scenario. This shows that entanglement distillation and nonlocality are different in general.

	\subsection{Connection to One-Way Distillation}
	
	It is worth pointing out a connection between \textit{one-way entanglement distillation} and $k$-copy steering. A state $\rho$ is said to be one-way distillable when it can be distilled via only one-way communication, say from Alice to Bob. 
	
	The connection goes as follows. The existence for a one-way distillation protocol for some state $\rho$ always implies that a maximally entangled state can also be obtained via a local filtering operation on one side only \cite{horodecki_review}. Since this class of operations cannot generate steering (see Lemma 1), it follows that \textit{all one-way distillable states are $k$-copy steerable}. 
	
	This allows us to detect $k$-copy steerable states using known criterion of one-way distillability, such as the one-way hashing protocol \cite{devetak05}. The latter says that any state $\rho$ that satisfies $S(\rho_{B})>S(\rho)$, with $S(\rho)$ being the Von-Neumann entropy of $\rho$, are one-way distillable. Hence, any state $\rho$ detected by the one-way hashing protocol is $k$-copy steerable from Alice to Bob.

\section{Discussion}

	In this paper we have explored steering in a scenario where the parties can perform joint measurements on many copies of a given entangled state. We gave a general and simple criterion to detect entangled states that exhibit $k$-copy steering. Notably, this criterion detects all entangled states in dimension $2 \times 2$ and $2 \times 3$. Moreover, we discussed several examples of super-activation of steering, some of which involving few copies and low dimensions. Finally, we connected this problem to entanglement distillation. 
	
	The natural question now would be: are all entangled states k-copy steerable? Although
we answered this question positively for the particular case of qubit-qubit and qubit-qutrit systems, there may exist some entangled qudit-qudit states that admits a LHS or LHV model even when an arbitrary number of copies is considered. This is similar in spirit to the still unresolved question of $k$-copy non-locality. While previous conjectures ruled out the possibility based on a connection to distillability \cite{peres96,pusey13}, the fact that undistillable states can exhibit steering \cite{moroder14} and even non-locality \cite{vertesi14} implies that the question is still essentially open.
	
	
	Another question is whether super-activation of steering necessarily implies the super-activation of nonlocality. For instance, could one find an entangled two-qubit state admitting an LHV model for any number of copies? The connection to entanglement distillability also deserves interest. Are all distillable states $k$-copy steerable (nonlocal)?

\section{Acknowledgements.}

 We acknowledge discussions Yeong-Cherng Liang, Carlos Palazuelos, Andreas Winter. MTQ and NB acknowledge financial support from the Swiss National Science Foundation (Starting grant DIAQ). MH acknowledges funding from the Swiss National Science Foundation (AMBIZIONE PZ00P2$\_$161351) and the Austrian Science Fund (FWF) through the START project Y879-N27.

\section*{Note added}
While completing this manuscript, we
became aware of the work of \cite{hsieh16} who independently
derived a sufficient condition for steerability in terms of
the fully entangled fraction and demonstrated the
super-activation of steering for the isotropic states.

\appendix

\section{More Details On The proof of  Theorem \ref{super}}
	We now reproduce a lemma first presented in Ref. \cite{horodecki_reduction} that we used for proving our main result.
\begin{lemma} \label{red2}
	Let $\rho$ be a bipartite state acting on $\mathbb{C}^{d}\otimes \mathbb{C}^{d}$  which	 $I\otimes \rho_B - \rho$ is not positive definite. There exists a local filter operation $F_B$ such that the filtered state
	\begin{align}
	\rho'=\frac{I\otimes F_B \rho I\otimes F_B^\dagger }{\tr(I\otimes F_B \rho I\otimes F_B^\dagger )}
	\end{align}
	satisfies	$\tr(\rho'\ketbra{\phi^+_d}{\phi^+_d)>1/d})>1/d$. 
\end{lemma}

\begin{proof}
	If the operator $I\otimes \rho_B - \rho$ is not positive definite, there exists a pure quantum state  $\ket{\psi}=\sum_{ij}\alpha_{ij}\ket{ij}$ such that
	\begin{align}
	\tr((I\otimes \rho_B - \rho)\ketbra{\psi}{\psi})<0
	\end{align}
	or equivalently, $\tr(\rho\ketbra{\psi}{\psi}))>\tr(\rho_B\psi_B)$, where $\psi_B=\tr_A(\ketbra{\psi}{\psi})$. We now define the local filter ${F_B:=\sqrt{d} \sum_{ij} \overline{\alpha_{ij}} \ketbra{i}{j}}$, with $\overline{\alpha_{ij}}$ being the complex conjugate of $\alpha_{ij}$. With that, we see that $\ket{\psi}= I\otimes F_B^\dagger\ket{\phi^+_d}$, and construct the filtered state is
	\begin{align}
	\rho':=\frac{I\otimes F_B \rho I\otimes F_B^\dagger }{\tr(I\otimes F_B \rho I\otimes F_B^\dagger )}.
	\end{align}
	Now, a straightforward calculation shows that
	\begin{align}
	& F_B^\dagger F_B= d \rho_B   \nonumber, \\
	&	\tr(I\otimes F_B \rho I\otimes F_B^\dagger) = d \tr (\rho_B\psi_B), \nonumber \\
	&	\tr( I\otimes F_B \rho I\otimes F_B^\dagger \ketbra{\phi^+_d}{\phi^+_d}) = \tr(\ketbra{\psi}{\psi} \rho),
	\end{align}
	and since we have $\tr(\rho\ketbra{\psi}{\psi})>\tr(\rho_B\psi_B)$ by hypothesis, it follows that 
	\begin{equation}
	\tr(\rho'\ketbra{\phi^+_d}{\phi^+_d})>1/d.
	\end{equation}
\end{proof}


\bibliographystyle{linksen}  
\bibliography{mtqbib.bib}        

\providecommand{\href}[2]{#2}\begingroup\raggedright\begin{thebibliography}{10}

\bibitem{bell64}
J.~S. Bell, ``{On the Einstein-Poldolsky-Rosen paradox},'' {\em Physics}
  {\bfseries 1}, 195--200 (1964).

\bibitem{brunner_review}
N.~{Brunner}, D.~{Cavalcanti}, S.~{Pironio}, V.~{Scarani}, and S.~{Wehner},
  ``{Bell nonlocality},''
  \href{http://dx.doi.org/10.1103/RevModPhys.86.419}{{\em Reviews of Modern
  Physics} {\bfseries 86}, 419--478 (2014)},
  \href{http://arxiv.org/abs/1303.2849}{{\ttfamily arXiv:1303.2849
  [quant-ph]}}.

\bibitem{acin07}
A.~{Ac{\'{\i}}n}, N.~{Brunner}, N.~{Gisin}, S.~{Massar}, S.~{Pironio}, and
  V.~{Scarani}, ``{Device-Independent Security of Quantum Cryptography against
  Collective Attacks},''
  \href{http://dx.doi.org/10.1103/PhysRevLett.98.230501}{{\em Physical Review
  Letters} {\bfseries 98}, 230501 (2007)},
  \href{http://arxiv.org/abs/quant-ph/0702152}{{\ttfamily quant-ph/0702152}}.

\bibitem{einstein35}
A.~Einstein, B.~Podolsky, and N.~Rosen, ``Can Quantum-Mechanical Description of
  Physical Reality Be Considered Complete?,''
  \href{http://dx.doi.org/10.1103/PhysRev.47.777}{{\em Phys. Rev.} {\bfseries
  47}, 777--780 (1935)}.

\bibitem{wiseman07}
H.~M. {Wiseman}, S.~J. {Jones}, and A.~C. {Doherty}, ``{Steering, Entanglement,
  Nonlocality, and the Einstein-Podolsky-Rosen Paradox},''
  \href{http://dx.doi.org/10.1103/PhysRevLett.98.140402}{{\em Phys. Rev. Lett.}
  {\bfseries 98}, 140402 (2007)},
  \href{http://arxiv.org/abs/quant-ph/0612147}{{\ttfamily quant-ph/0612147}}.

\bibitem{branciard11b}
C.~{Branciard}, E.~G. {Cavalcanti}, S.~P. {Walborn}, V.~{Scarani}, and H.~M.
  {Wiseman}, ``{One-sided device-independent quantum key distribution:
  Security, feasibility, and the connection with steering},''
  \href{http://dx.doi.org/10.1103/PhysRevA.85.010301}{{\em {Phys. Rev.~A}}
  {\bfseries 85}, 010301 (2012)},
  \href{http://arxiv.org/abs/1109.1435}{{\ttfamily arXiv:1109.1435
  [quant-ph]}}.

\bibitem{werner89}
R.~F. Werner, ``Quantum states with Einstein-Podolsky-Rosen correlations
  admitting a hidden-variable model,''
  \href{http://dx.doi.org/10.1103/PhysRevA.40.4277}{{\em Phys. Rev. A}
  {\bfseries 40}, 4277--4281 (1989)}.

\bibitem{barrett02}
J.~{Barrett}, ``{Nonsequential positive-operator-valued measurements on
  entangled mixed states do not always violate a Bell inequality},''
  \href{http://dx.doi.org/10.1103/PhysRevA.65.042302}{{\em Phys. Rev.~A}
  {\bfseries 65}, 042302 (2002)}.

\bibitem{quintino15}
M.~T. {Quintino}, T.~{V{\'e}rtesi}, D.~{Cavalcanti}, R.~{Augusiak},
  M.~{Demianowicz}, A.~{Ac{\'{\i}}n}, and N.~{Brunner}, ``{Inequivalence of
  entanglement, steering, and Bell nonlocality for general measurements},''
  \href{http://dx.doi.org/10.1103/PhysRevA.92.032107}{{\em Phys. Rev. A}
  {\bfseries 92}, 032107 (2015)},
  \href{http://arxiv.org/abs/1501.03332}{{\ttfamily arXiv:1501.03332
  [quant-ph]}}.

\bibitem{peres96}
A.~{Peres}, ``{Collective tests for quantum nonlocality},''
  \href{http://dx.doi.org/10.1103/PhysRevA.54.2685}{{\em Phys. Rev.~A}
  {\bfseries 54}, 2685--2689 (1996)}.

\bibitem{masanes06}
L.~{Masanes}, ``{Asymptotic Violation of Bell Inequalities and
  Distillability},''
  \href{http://dx.doi.org/10.1103/PhysRevLett.97.050503}{{\em Physical Review
  Letters} {\bfseries 97}, 050503 (2006)},
  \href{http://arxiv.org/abs/quant-ph/0512153}{{\ttfamily quant-ph/0512153}}.

\bibitem{liang06}
Y.-C. {Liang} and A.~C. {Doherty}, ``{Better Bell-inequality violation by
  collective measurements},''
  \href{http://dx.doi.org/10.1103/PhysRevA.73.052116}{{\em Physical Review A}
  {\bfseries 73}, 052116 (2006)},
  \href{http://arxiv.org/abs/quant-ph/0604045}{{\ttfamily quant-ph/0604045}}.

\bibitem{chsh69}
J.~F. Clauser, M.~A. Horne, A.~Shimony, and R.~A. Holt, ``Proposed Experiment
  to Test Local Hidden-Variable Theories,''
  \href{http://dx.doi.org/10.1103/PhysRevLett.23.880}{{\em Phys. Rev. Lett.}
  {\bfseries 23}, 880--884 (1969)}.

\bibitem{vertesi11}
M.~{Navascu{\'e}s} and T.~{V{\'e}rtesi}, ``{Activation of Nonlocal Quantum
  Resources},'' \href{http://dx.doi.org/10.1103/PhysRevLett.106.060403}{{\em
  Physical Review Letters} {\bfseries 106}, 060403 (2011)},
  \href{http://arxiv.org/abs/1010.5191}{{\ttfamily arXiv:1010.5191
  [quant-ph]}}.

\bibitem{caban15}
P.~Caban, A.~Molenda, and K.~Trzci\ifmmode~\acute{n}\else \'{n}\fi{}ska,
  ``Activation of the violation of the Svetlichny inequality,''
  \href{http://dx.doi.org/10.1103/PhysRevA.92.032119}{{\em Phys. Rev. A}
  {\bfseries 92}, 032119 (2015)}.
  \url{http://link.aps.org/doi/10.1103/PhysRevA.92.032119}.

\bibitem{paul15}
B.~{Paul}, K.~{Mukherjee}, and D.~{Sarkar}, ``{Activation of Genuine Tripartite
  Nonlocality},'' {\em ArXiv e-prints} (2015),
  \href{http://arxiv.org/abs/1512.08138}{{\ttfamily "arXiv":1512.08138
  [quant-ph]}}.

\bibitem{forster09}
M.~{Forster}, S.~{Winkler}, and S.~{Wolf}, ``{Distilling Nonlocality},''
  \href{http://dx.doi.org/10.1103/PhysRevLett.102.120401}{{\em Physical Review
  Letters} {\bfseries 102}, 120401 (2009)},
  \href{http://arxiv.org/abs/0809.3173}{{\ttfamily arXiv:0809.3173
  [quant-ph]}}.

\bibitem{BS}
N.~{Brunner} and P.~{Skrzypczyk}, ``{Nonlocality Distillation and Postquantum
  Theories with Trivial Communication Complexity},''
  \href{http://dx.doi.org/10.1103/PhysRevLett.102.160403}{{\em Physical Review
  Letters} {\bfseries 102}, 160403 (2009)},
  \href{http://arxiv.org/abs/0901.4070}{{\ttfamily arXiv:0901.4070
  [quant-ph]}}.

\bibitem{brunner10}
N.~{Brunner}, D.~{Cavalcanti}, A.~{Salles}, and P.~{Skrzypczyk}, ``{Bound
  Nonlocality and Activation},''
  \href{http://dx.doi.org/10.1103/PhysRevLett.106.020402}{{\em Physical Review
  Letters} {\bfseries 106}, 020402 (2011)},
  \href{http://arxiv.org/abs/1009.4207}{{\ttfamily arXiv:1009.4207
  [quant-ph]}}.

\bibitem{palazuelos12}
C.~{Palazuelos}, ``{Superactivation of Quantum Nonlocality},''
  \href{http://dx.doi.org/10.1103/PhysRevLett.109.190401}{{\em Physical Review
  Letters} {\bfseries 109}, 190401 (2012)}.

\bibitem{junge09}
M.~{Junge}, C.~{Palazuelos}, D.~{P{\'e}rez-Garc{\'{\i}}a}, I.~{Villanueva}, and
  M.~M. {Wolf}, ``{Unbounded Violations of Bipartite Bell Inequalities via
  Operator Space Theory},''
  \href{http://dx.doi.org/10.1007/s00220-010-1125-5}{{\em Communications in
  Mathematical Physics} {\bfseries 300}, 715--739 (2010)},
  \href{http://arxiv.org/abs/0910.4228}{{\ttfamily arXiv:0910.4228
  [quant-ph]}}.

\bibitem{KV}
S.~A. {Khot} and N.~K. {Vishnoi}, ``{The Unique Games Conjecture, Integrality
  Gap for Cut Problems and Embeddability of Negative Type Metrics into
  $\ell\_1$},'' {\em ArXiv e-prints} (2013),
  \href{http://arxiv.org/abs/1305.4581}{{\ttfamily arXiv:1305.4581 [cs.CC]}}.

\bibitem{buhrman10}
H.~{Buhrman}, O.~{Regev}, G.~{Scarpa}, and R.~{de Wolf}, ``{Near-Optimal and
  Explicit Bell Inequality Violations},'' {\em Proceedings of IEEE Complexity}
  (2011), \href{http://arxiv.org/abs/1012.5043}{{\ttfamily arXiv:1012.5043
  [quant-ph]}}.

\bibitem{cavalcanti12}
D.~{Cavalcanti}, A.~{Ac{\'{\i}}n}, N.~{Brunner}, and T.~{V{\'e}rtesi}, ``{All
  quantum states useful for teleportation are nonlocal resources},''
  \href{http://dx.doi.org/10.1103/PhysRevA.87.042104}{{\em Phys. Rev.~A}
  {\bfseries 87}, 042104 (2013)},
  \href{http://arxiv.org/abs/1207.5485}{{\ttfamily arXiv:1207.5485
  [quant-ph]}}.

\bibitem{horodecki98}
M.~Horodecki, P.~Horodecki, and R.~Horodecki, ``General teleportation channel,
  singlet fraction, and quasidistillation,''
  \href{http://dx.doi.org/10.1103/PhysRevA.60.1888}{{\em Phys. Rev. A}
  {\bfseries 60}, 1888--1898 (1999)},
  \href{http://arxiv.org/abs/quant-ph/9807091}{{\ttfamily
  arXiv:quant-ph/9807091}}.

\bibitem{horodecki_reduction}
M.~{Horodecki} and P.~{Horodecki}, ``{Reduction criterion of separability and
  limits for a class of protocols of entanglement distillation},'' {\em eprint
  arXiv:quant-ph/9708015} (1997),
  \href{http://arxiv.org/abs/quant-ph/9708015}{{\ttfamily quant-ph/9708015}}.

\bibitem{bowles13}
J.~{Bowles}, M.~T. {Quintino}, and N.~{Brunner}, ``{Certifying the Dimension of
  Classical and Quantum Systems in a Prepare-and-Measure Scenario with
  Independent Devices},''
  \href{http://dx.doi.org/10.1103/PhysRevLett.112.140407}{{\em Physical Review
  Letters} {\bfseries 112}, 140407 (2014)},
  \href{http://arxiv.org/abs/1311.1525}{{\ttfamily arXiv:1311.1525
  [quant-ph]}}.

\bibitem{sania14}
S.~{Jevtic}, M.~J.~W. {Hall}, M.~R. {Anderson}, M.~{Zwierz}, and H.~M.
  {Wiseman}, ``{Einstein-Podolsky-Rosen steering and the steering ellipsoid},''
  \href{http://dx.doi.org/10.1364/JOSAB.32.000A40}{{\em Journal of the Optical
  Society of America B Optical Physics} {\bfseries 32}, A40 (2015)},
  \href{http://arxiv.org/abs/1411.1517}{{\ttfamily arXiv:1411.1517
  [quant-ph]}}.

\bibitem{bowles16}
J.~{Bowles}, J.~{Francfort}, M.~{Fillettaz}, F.~{Hirsch}, and N.~{Brunner},
  ``{Genuinely Multipartite Entangled Quantum States with Fully Local Hidden
  Variable Models and Hidden Multipartite Nonlocality},''
  \href{http://dx.doi.org/10.1103/PhysRevLett.116.130401}{{\em Physical Review
  Letters} {\bfseries 116}, 130401 (2016)},
  \href{http://arxiv.org/abs/1511.08401}{{\ttfamily arXiv:1511.08401
  [quant-ph]}}.

\bibitem{hirsch15}
F.~{Hirsch}, M.~{T{\'u}lio Quintino}, T.~{V{\'e}rtesi}, M.~F. {Pusey}, and
  N.~{Brunner}, ``{Algorithmic construction of local hidden variable models for
  entangled quantum states},'' {\em ArXiv e-prints} (2015),
  \href{http://arxiv.org/abs/1512.00262}{{\ttfamily arXiv:1512.00262
  [quant-ph]}}.

\bibitem{cavalcanti15}
D.~{Cavalcanti}, L.~{Guerini}, R.~{Rabelo}, and P.~{Skrzypczyk}, ``{General
  method for constructing local-hidden-variable models for multiqubit entangled
  states},'' {\em ArXiv e-prints} (2015),
  \href{http://arxiv.org/abs/1512.00277}{{\ttfamily arXiv:1512.00277
  [quant-ph]}}.

\bibitem{nguyen16}
H.~{Chau Nguyen} and T.~{Vu}, ``{Necessary and sufficient condition for
  steerability of two-qubit states by the geometry of steering outcomes},''
  \href{http://dx.doi.org/10.1209/0295-5075/115/10003}{{\em EPL (Europhysics
  Letters)} {\bfseries 115}, 10003 (2016)},
  \href{http://arxiv.org/abs/1604.03815}{{\ttfamily arXiv:1604.03815
  [quant-ph]}}.

\bibitem{cavalcanti09}
E.~G. {Cavalcanti}, S.~J. {Jones}, H.~M. {Wiseman}, and M.~D. {Reid},
  ``{Experimental criteria for steering and the Einstein-Podolsky-Rosen
  paradox},'' \href{http://dx.doi.org/10.1103/PhysRevA.80.032112}{{\em Phys.
  Rev. A} {\bfseries 80}, 032112 (2009)},
  \href{http://arxiv.org/abs/0907.1109}{{\ttfamily arXiv:0907.1109
  [quant-ph]}}.

\bibitem{augusiak_review}
R.~Augusiak, M.~Demianowicz, and A.~Acín, ``Local hidden–variable models for
  entangled quantum states,'' {\em Journal of Physics A: Mathematical and
  Theoretical} {\bfseries 47}, 424002 (2014).
  \url{http://stacks.iop.org/1751-8121/47/i=42/a=424002}.

\bibitem{bowles14}
J.~{Bowles}, T.~{V{\'e}rtesi}, M.~T. {Quintino}, and N.~{Brunner}, ``{One-way
  Einstein-Podolsky-Rosen Steering},''
  \href{http://dx.doi.org/10.1103/PhysRevLett.112.200402}{{\em Physical Review
  Letters} {\bfseries 112}, 200402 (2014)},
  \href{http://arxiv.org/abs/1402.3607}{{\ttfamily arXiv:1402.3607
  [quant-ph]}}.

\bibitem{guhne08}
O.~{G{\"u}hne} and G.~{T{\'o}th}, ``{Entanglement detection},''
  \href{http://dx.doi.org/10.1016/j.physrep.2009.02.004}{{\em Physics Reports}
  {\bfseries 474}, 1--75 (2009)},
  \href{http://arxiv.org/abs/0811.2803}{{\ttfamily arXiv:0811.2803
  [quant-ph]}}.

\bibitem{aolita14}
R.~{Gallego} and L.~{Aolita}, ``{Resource Theory of Steering},''
  \href{http://dx.doi.org/10.1103/PhysRevX.5.041008}{{\em Physical Review X}
  {\bfseries 5}, 041008 (2015)},
  \href{http://arxiv.org/abs/1409.5804}{{\ttfamily arXiv:1409.5804
  [quant-ph]}}.

\bibitem{horodecki_ppt}
M.~{Horodecki}, P.~{Horodecki}, and R.~{Horodecki}, ``{Separability of mixed
  states: necessary and sufficient conditions},''
  \href{http://dx.doi.org/10.1016/S0375-9601(96)00706-2}{{\em Physics Letters
  A} {\bfseries 223}, 1--8 (1996)},
  \href{http://arxiv.org/abs/quant-ph/9605038}{{\ttfamily quant-ph/9605038}}.

\bibitem{almeida07}
M.~L. {Almeida}, S.~{Pironio}, J.~{Barrett}, G.~{T{\'o}th}, and
  A.~{Ac{\'{\i}}n}, ``{Noise Robustness of the Nonlocality of Entangled Quantum
  States},'' \href{http://dx.doi.org/10.1103/PhysRevLett.99.040403}{{\em Phys.
  Rev. Lett.} {\bfseries 99}, 040403 (2007)}.

\bibitem{bennett96}
C.~H. {Bennett}, G.~{Brassard}, S.~{Popescu}, B.~{Schumacher}, J.~A. {Smolin},
  and W.~K. {Wootters}, ``{Purification of Noisy Entanglement and Faithful
  Teleportation via Noisy Channels},''
  \href{http://dx.doi.org/10.1103/PhysRevLett.76.722}{{\em Physical Review
  Letters} {\bfseries 76}, 722--725 (1996)},
  \href{http://arxiv.org/abs/quant-ph/9511027}{{\ttfamily quant-ph/9511027}}.

\bibitem{horodecki_review}
R.~{Horodecki}, P.~{Horodecki}, M.~{Horodecki}, and K.~{Horodecki}, ``{Quantum
  entanglement},'' \href{http://dx.doi.org/10.1103/RevModPhys.81.865}{{\em
  Reviews of Modern Physics} {\bfseries 81}, 865--942 (2009)},
  \href{http://arxiv.org/abs/quant-ph/0702225}{{\ttfamily quant-ph/0702225}}.

\bibitem{pusey13}
M.~F. {Pusey}, ``{Negativity and steering: A stronger Peres conjecture},''
  \href{http://dx.doi.org/10.1103/PhysRevA.88.032313}{{\em Phys. Rev.~A}
  {\bfseries 88}, 032313 (2013)}.

\bibitem{moroder14}
T.~{Moroder}, O.~{Gittsovich}, M.~{Huber}, and O.~{G{\"u}hne}, ``{Steering
  Bound Entangled States: A Counterexample to the Stronger Peres Conjecture},''
  \href{http://dx.doi.org/10.1103/PhysRevLett.113.050404}{{\em Phys. Rev.
  Lett.} {\bfseries 113}, 050404 (2014)},
  \href{http://arxiv.org/abs/1405.0262}{{\ttfamily arXiv:1405.0262
  [quant-ph]}}.

\bibitem{vertesi14}
T.~{V{\'e}rtesi} and N.~{Brunner}, ``{Disproving the Peres conjecture by
  showing Bell nonlocality from bound entanglement},''
  \href{http://dx.doi.org/10.1038/ncomms6297}{{\em Nature Communications}
  {\bfseries 5}, 5297 (2014)}, \href{http://arxiv.org/abs/1405.4502}{{\ttfamily
  arXiv:1405.4502 [quant-ph]}}.

\bibitem{devetak05}
I.~{Devetak} and A.~{Winter}, ``{Distillation of secret key and entanglement
  from quantum states},'' \href{http://dx.doi.org/10.1098/rspa.2004.1372}{{\em
  Proceedings of the Royal Society of London Series A} {\bfseries 461},
  207--235 (2005)}, \href{http://arxiv.org/abs/quant-ph/0306078}{{\ttfamily
  quant-ph/0306078}}.

\bibitem{hsieh16}
C.-Y. {Hsieh}, Y.-C. {Liang}, and R.-K. {Lee}, ``{Quantum steerability:
  characterization, quantification, superactivation and unbounded
  amplification},'' {\em ArXiv e-prints} (2016),
  \href{http://arxiv.org/abs/1609.07581}{{\ttfamily arXiv:1609.07581
  [quant-ph]}}.

\end{thebibliography}\endgroup

\end{document}